\documentclass{amsart}
\usepackage{amssymb}
\usepackage{latexsym}
\usepackage{amsmath}
\usepackage{euscript}
\usepackage{color}
\usepackage{graphics}
\usepackage[all]{xy}

\usepackage[leqno]{amsmath}

\usepackage{bbm}

\usepackage{dsfont}

\usepackage{graphicx}

      \def\dR{{\mathbb R}}

   \def\dZ{{\mathbb Z}}

\newcommand{\be}{\begin{equation}}
\newcommand{\ee}{\end{equation}}
\newcommand{\ba}{\begin{eqnarray}}
\newcommand{\ea}{\end{eqnarray}}
\newcommand{\baa}{\begin{eqnarray*}}
\newcommand{\eaa}{\end{eqnarray*}}
\newcommand{\bb}{\begin{thebibliography}}
\newcommand{\eb}{\end{thebibliography}}

\newcommand{\bi}[1]{\bibitem{#1}}
\newcommand{\lab}[1]{\label{#1}}
\newcommand{\re}[1]{(\ref{#1})}

\newcounter{my}
\newcommand{\he}%
   {\stepcounter{equation}\setcounter{my}%
   {\value{equation}}\setcounter{equation}0%
   }%
\newcommand{\she}%
   {\setcounter{equation}{\value{my}}%
    }%

\renewcommand\t{\tilde}
\def\sign{\operatorname{sign}}

\newtheorem{theorem}{Theorem}[section]

\newtheorem{lemma}[theorem]{Lemma}
\newtheorem{definition}[theorem]{Definition}
\theoremstyle{definition}

\newtheorem{remark}[theorem]{Remark}

\numberwithin{equation}{section}

\begin{document}

\title[Bannai-Ito polynomials and dressing chains]
{Bannai-Ito polynomials and dressing chains}

\author{Maxim Derevyagin}
\author{Satoshi Tsujimoto}
\author{Luc Vinet}
\author{Alexei Zhedanov}

\address{Department of Mathematics MA 4-2, Technische Universit\"at Berlin, Strasse des 17. 
Juni 136, D-10623 Berlin, Germany}

\address{Department of Applied Mathematics and Physics, Graduate School of Informatics, Kyoto University,
Sakyo-ku, Kyoto 606-8501, Japan}

\address{Centre de recherches math\'ematiques,
Universit\'e de Montr\'eal, P.O. Box 6128, Centre-ville Station,
Montr\'eal (Qu\'ebec), H3C 3J7,Canada}

\address{Donetsk Institute for Physics and Technology\\
83114 Donetsk, Ukraine \\}

\begin{abstract}
Schur-Delsarte-Genin (SDG) maps and Bannai-Ito polynomials are studied. 
SDG maps are related to dressing chains determined by quadratic algebras. 
The Bannai-Ito polynomials and their kernel polynomials -- the complementary Bannai-Ito
polynomials -- are shown to arise in the framework of the SDG maps.
\end{abstract}


\maketitle

\section{Introduction}

It is well known that periodic dressing chains for Schr\"odinger operators can be used to characterize
finite-gap potentials~\cite{VSh}. There were also some attempts to adapt the scheme for tridiagonal operators
and orthogonal polynomials~\cite{GH96},~\cite{SZ95},~\cite{SZ97}. Naturally, this research leads to the following question: what families of orthogonal polynomials
can be described via periodic dressing chains for tridiagonal operators?  

In this paper we pursue the exploration of this question and study the interplay between dressing chains
for tridiagonal matrices and known families of orthogonal polynomials. More precisely, the goal of the present paper
is twofold. First, we propose a way to define Darboux transformations for bilinear pencils of the form
\[
A-\lambda B-x I,
\] 
where $A$, $B$ are tridiagonal semi-infinite matrices and $I$ is the identity operator. 
One shall see that these transformations lead to dressing chains related to quadratic algebras. 
Remarkably, the pencils $A-\lambda B-x I$ appear in different contexts~\cite{BMW},~\cite{BW10},~\cite{ER}
(see also~\cite{Stanton},~\cite{ST98}, and~\cite{Uchiyama}, where similar objects arise). Second, we construct as an
explicit example, the dressing chain related to the Bannai-Ito polynomials~\cite{BI},~\cite{TVZh12}. 
Its presentation in an open form will require to set the Bannai-Ito polynomials
in the framework of Schur-Delsarte-Genin maps. 

The paper is organized as follows. In Section 2 we define the generic Darboux transformations for tridiagonal matrices
and identify the relation between 1-periodic dressing chains and quadratic algebras. Examples of dressing chains
generated by orthogonal polynomials on the unit circle are given in Section~3. The next section is concerned with
Schur-Delsarte-Genin maps, which are essentially the bridges between dressing chains  and orthogonal
polynomials on the unit circle. The rest of the paper deals with the Bannai-Ito polynomials as an explicit
example. Namely, in Section~5 we transform the complementary Bannai-Ito polynomials
into polynomials on the unit circle. These new polynomials turn out to be related to symmetrized Racah-Wilson
polynomials, as shown in Section~6. Finally, in Section~7, we discuss the role of the Bannai-Ito polynomials
in Schur-Delsarte-Genin maps.

\section{Generic Darboux transformations}

In this section we give the definition of the generic Darboux transformations and identify quadratic algebras 
of tridiagonal operators as dressing chains. 

Before giving the general abstract idea, let us briefly recall the interconnection between discrete Darboux 
transformations and dynamical systems. Consider a finite dimensional case and let $J$ be a finite Jacobi matrix, that is,
\[
 J=\begin{pmatrix}
{a}_{0}   &{b}_{0}&       &\\
{b}_{0}   & {a}_{1}   & \ddots&\\
        & \ddots   & \ddots &{b}_{N-1}\\
&       &{b}_{N-1} &{a}_N\\
\end{pmatrix},
\]
where $a_j\in\dR$ and $b_j>0$. In addition, assume that $J$ is positive definite. Then $J$ admits the Cholesky
decomposition
\[
 J=LL^*,
\]
where $L$ is a lower triangular matrix and $\det L\ne 0$. One can see that in this case $L$ is bidiagonal.
Recall~\cite{BM},~\cite{ZhS} that {\it the Darboux transformation of $J$} is defined as follows
\[
 J=LL^*\mapsto \t J=L^*L.
\]
Clearly, $\t J$ is also a Jacobi matrix and the spectrum $\sigma(\t J)$ of $\t J$ coincides with the spectrum $\sigma(J)$ of J.
Indeed, from the definition we get that
\[
 \t J L^*=L^*J.
\]
Thus $(\t J-\lambda I)L^*=L^*(J-\lambda I)$, where $\lambda$ is a complex number and 
$I$ is the corresponding identity operator. Hence, $\sigma(\t J)=\sigma(J)$ since $\det L\ne 0$. Furthermore,
the eigenvectors $\t f_j$ and $f_j$ of $\t J$ and $J$, respectively, are related as follows
\[
 \t f_j=L^*f_j.
\]
 
It follows that the iterates of the Darboux transformation applied to $J$ are isospectral and related
to an isospectral flow~\cite{W84}.

Now we are in a position to give a more formal definition of the Darboux transformation of tridagonal semi-infinite
matrices
\begin{equation}\label{JacForm}
 J= \begin{pmatrix}
 \alpha_0 & \beta_0 &  &    \\
 \gamma_0 & \alpha_1 &  \beta_1&   \\
  &  \gamma_1 & \alpha_2 &   \beta_2  &  \\
     &  & \gamma_2 &  \ddots  \\
     &&&\\
\end{pmatrix}, \quad |\beta_j|+|\gamma_j|\ne 0, \quad  j\in\dZ_+, 
\end{equation}
 which are the discrete analog of the Schr\"odinger operator and appear in the theory of orthogonal polynomials and
continued fractions~\cite{Simon},~\cite{SZ}.

\begin{definition}\label{genDarboux}
Let $J$ be of the form~\eqref{JacForm} and let $\psi$
be a column vector defined by the condition
\be J \psi =0.
\lab{J_psi} \ee 
Also, let $D$ be a semi-infinite nonzero matrix such that the product $D\psi$ makes sense. Introduce a new vector 
\be \t \psi = D \psi. \lab{D_psi} \ee 
Now, suppose that there exist semi-infinite matrices
$K,\t J$  that satisfy the intertwining relation 
\be \t J D = K J. \lab{JDKJ} \ee 
Then the operator $\t J$ is called the Darboux transformation of $J$ and it is easy to see that 
 \be \t J \t \psi =0.
\lab{tJtpsi} \ee
\end{definition}
 We will call this scheme the generic Darboux
scheme and formula \re{D_psi} describes the Darboux transformation
from the vector $\psi$ to the vector $\t \psi$. Formula \re{JDKJ} is an
algebraic relation between the operators $J$ and $\t J$.

Obviously, the generic Darboux scheme gives a family of Darboux transformations 
depending on the choice of $D$. For example, for any two Jacobi matrices $J$ and $\t J$ one can trivially take
$D=cJ$ and $K=c\t J$, where $c$ is a nonzero complex number. Therefore, any Jacobi matrix $\t J$ 
can be considered as the generic Darboux transformation of the given matrix $J$. Clearly, 
this is not too helpful and generally leads to trivial results. However, there are many nontrivial explicit examples
which can be obtained by putting some requirements on $D$, $K$ and by letting the matrices 
$J$, $\t J$ depend on parameters. Let us consider the simplest such case: 
\begin{equation}\label{ClJacPen}
 J(\lambda) = A-\lambda I, 
\end{equation}
where $\lambda$ is a parameter and $A$ is a monic Jacobi matrix with real entries that
do not depend on $\lambda$. In this case, the entries of $\psi$ are the corresponding 
orthogonal polynomials evaluated at $\lambda$. Let us fix $\lambda$ and suppose that $J(\lambda)$ admits
the $LU$-decomposition, that is, $J(\lambda)=LU$, where $L$ and $U$ are lower and upper triangular matrices.
Then introducing $\t J(\lambda)=\t A-\lambda I=UL$ we see that~\eqref{JDKJ} is satisfied
for $K=D=U$. Besides, the entries of $\t \psi=D\psi$ are related to the orthogonal polynomials 
corresponding to~$\t A$~\cite{BM},~\cite{ZhS} (see also~\cite{DD11} for some non-classical cases).

Another example can be obtained by considering the following form of the matrices $J$ and $\t J$:
 \be J(\lambda) = A -
\lambda B, \quad \t J(\lambda) = \t A - \lambda \t B, \lab{J_AB}
\ee 
where $A$ and $B$ are some tridiagonal matrices. Assuming that the operators
$D$ and $K$ do not depend on $\lambda$ we obtain that \re{JDKJ} is
equivalent to the following  condition~\cite{WE89},~\cite{ZheBRF} 
\be \t A D = K A, \quad \t B D = K B.
\lab{tAB_DK} \ee 
Actually,  this entails the generalized eigenvalue problem 
\be A \psi = \lambda B \psi. \lab{Apsi_Bpsi}
\ee 
Note that such problems appear in mechanics~\cite{Gladwell}, in the theory of 
biorthogonal rational functions~\cite{ZheBRF}, and in interpolation problems~\cite{DZh}.
In this case, the Darboux transformation \re{D_psi} gives a new eigenvector
$\t \psi$ of the new generalized eigenvalue problem \be \t A \t
\psi = \lambda \t B \t \psi. \lab{tGEVP} \ee

Finally, we are ready to discuss the generalization that is of interest in this paper. 
Namely, consider the following form of $J$: 
\be J(\lambda,x)= A -\lambda B -x I, 
\lab{JtJ_xl} \ee 
where the tridiagonal matrices $A$ and $B$ are independent of the real parameters $\lambda$ and $x$. We see that 
when one fixes $\lambda$ the problem for $J(\lambda, x)$ reduces to~\eqref{ClJacPen} and
 $J(\lambda, x)$ becomes~\eqref{J_AB} if $x$ is fixed.

Note that the initial  condition for $\psi$ 
\be (A - \lambda B - x)\psi=0
\lab{ABxl} \ee 
is a typical 2-parameter problem
\cite{Atkinson}, \cite{Sleeman}. Our purpose is to describe 1-periodic dressing chains~\cite{VSh} related to such problems.
Before giving the precise definition, let us specify the definition of the Darboux transformation of $J$
of the form~\eqref{JtJ_xl}. In this case, the Darboux transformation is obtained from the generic Darboux scheme
(see Definition~\ref{genDarboux}),
where we want the  intertwining (Darboux) operators $D$ and $K$ to have the following form 
 \be D = r_1 A + r_2 B
+ r_0 I, \quad  K = s_1 A + s_2 B + s_0 I, \lab{DK_lin} \ee where $r_i, s_i$ are some real numbers. 
As always, this definition
can lead to nontrivial results if the matrix $D$ is not proportional to $J$. 
\begin{definition}
Let $J$ be a tridiagonal matrix given by~\eqref{JtJ_xl}. We say that $J$ generates
a 1-periodic dressing chain if the Darboux transformation~\eqref{JDKJ} subject 
to the Ansatz~\eqref{DK_lin} gives a tridiagonal matrix of the form 
\begin{equation}\lab{JtJ_x2}
\t J=\widetilde{J(\lambda,x)}= A - \t \lambda B - \t x I,
\end{equation}
where $\t x$, $\t\lambda$ are some numbers which can be considered as the transformation
of $x$ and $\lambda$. Notice that $\t J$ generates a new 2-parameter problem 
\be (A - \t
\lambda B - \t x) \t \psi=0 \lab{D_EVP} \ee 
similar to~\eqref{ABxl} but with transformed
parameters $\t \lambda$ and  $\t x$.
\end{definition}

In other words, it means that the Darboux transformation applied to $J(\lambda,x)$ belongs
to the linear span generated by $A$, $B$, and $I$.  The following statement gives 
a description of the 1-periodic dressing chains related to $J(\lambda,x)$.  

\begin{theorem}\label{QuadrAlgebra}
Let the matrix $J$ be given by~\eqref{JtJ_xl}. If $J$ generates a 1-periodic dressing 
chain then the following quadratic relation 
\be \xi_1 A^2 + \xi_2 B^2 + \xi_3 AB + \xi_4 BA
+ \eta_1 A + \eta_2 B + \zeta I =0, \lab{gen_quad} \ee 
has a solution in $\xi_i$, $\eta_i$, and $\zeta$. This dressing chain is nontrivial
if the condition
\[
 |\xi_1|+|\xi_2|+|\xi_3|+|\xi_4|+|\eta_1|+|\eta_2| +|\zeta|\ne 0
\]
is satisfied for the solution.
\end{theorem}
\begin{proof}
 The proof is immediate by making the substitution~\eqref{JtJ_xl},~\eqref{JtJ_x2} and~\eqref{DK_lin} to~\eqref{JDKJ}.
\end{proof}

\begin{remark}
It should be mentioned that martingale orthogonal polynomials satisfy~\eqref{ABxl}~\cite{BMW}. Furthermore,
the quadratic algebraic relations~\eqref{gen_quad} were considered in~\cite{BMW} and~\cite{BW10} 
in connection with quadratic harnesses and martingale orthogonal polynomials. Actually, it was shown that
the Jacobi matrices corresponding to the Askey-Wilson polynomials satisfy~\eqref{gen_quad}. Moreover, the quadratic 
algebra~\eqref{gen_quad} is related to an exactly solvable box-and-ball kinetic model in physics~\cite{ER}.
\end{remark}

Note that the algebra~\re{gen_quad} remains the same under the generic affine
transformation $A \to \alpha_1 A + \alpha_2 B + \alpha_0$, 
$B \to \beta_1 A + \beta_2 B + \beta_0$. Using this observation we can put the quadratic algebra
\re{gen_quad} in canonical forms. In fact, such a program  was realized for self-adjoint operators in~\cite{OS}.
We are not going to do this in the present paper. Here we shall give a particular example of 
the quadratic algebra~\re{gen_quad} that is related to
Schur-Delsarte-Genin maps~\cite{DVZh} and not covered in~\cite{BMW} and~\cite{BW10}.
One of the main goals of this paper is hence to show that the Jacobi matrices corresponding to the big -1 Jacobi polynomials
and the Bannai-Ito polynomials satisfy the relation~\eqref{gen_quad} thereby extending the class of polynomials
related to quadratic algebras.

\section{Orthogonal polynomials on the unit circle and quadratic algebras}

In this section we introduce a 2-parameter pencil and show that it is related to quadratic algebras and therefore to
1-periodic dressing chains. 

Let us start by recalling that monic orthogonal polynomials $\Phi_n(z)=z^n + O(z^{n-1})$ on the unit circle 
satisfy the recurrence relation
\be
\Phi_{n+1}(z) = z\Phi_n(z) - \bar a_{n} \Phi_n^*(z),
\lab{Sz_rec} \ee
with the initial condition $\Phi_0(z) =1$; here the polynomials $\Phi_n^*(z)$ are defined as follows
$$
\Phi_n^*(z) = z^n \overline{\Phi_n(1/\bar z)}
$$
(see~\cite{Simon}).
The recurrence coefficients $a_n=-\overline{\Phi_{n+1}(0)}$ are called
reflection parameters (sometimes also referred to as the Schur, Geronimus,
Verblunsky, ... parameters). 
It follows that $|a_n|<1$ for $n=0,1,2,\dots$. Conversely (see~\cite{Simon})
polynomials satisfying~\eqref{Sz_rec} with this condition on $a_n$ are orthogonal.

{\it In what follows we will consider only the case where the $a_n$ are real}.

For convenience, we always assume that
\be a_{-1}=-1 \lab{a_-1}.
\ee

The dual recurrence relation is
\be
\Phi_{n+1}^*(z) = \Phi_n^*(z) - a_n z \Phi_n(z) \lab{dual_rec_Phi}. \ee

It is standard to introduce the parameters
$$
r_n = \sqrt{1-a_n^2}
$$
(in our case, the square root is assumed to be positive).

Let us introduce the block-diagonal matrices
\vspace{5mm}
\be
 L =
 \begin{pmatrix}
  a_{0} & r_{0} &  &    \\
  r_{0} & -a_{0} &  &   \\
   &  &              a_{2} & r_{2}  \\
   &  &              r_{2} & -a_{2}  \\
  &  &    & &          a_{4} & r_{4}  \\
   &  &    & &         r_{4} & - a_{4}  \\
&   &  & & & & \ddots  \\
 \end{pmatrix}
\lab{L_def} \ee
and 
\vspace{5mm}
\be
 M =
 \begin{pmatrix}
1 \\
& a_{1} & r_{1} &  &    \\
  &r_{1} & -a_{1} &  &   \\
   & &  &               a_{3} & r_{3}  \\
   & & &              r_{3} & -a_{3}  \\
  &  & &   & &           a_{5} & r_{5}  \\
   &  & &   & &         r_{5} & - a_{5}  \\
&   &  & & & & & \ddots  \\
 \end{pmatrix}.
\lab{M_def} \ee
Both $L$ and $M$ are block-diagonal unitary matrices. Define also the 5-diagonal unitary matrix $U$ \cite{Watkins}, \cite{Simon}
\be
U=LM \lab{def_U} \ee
(which is called the "Zigzag matrix" in \cite{Watkins} and the "CMV matrix" in \cite{Simon}).

Next, consider a linear pencil of matrices
\be
K(\lambda) = L + \lambda M \lab{lin_K} \ee
with an arbitrary parameter $\lambda$.
So, the matrix $K(\lambda)$ is a nondegenerate symmetric
tri-diagonal matrix
\[
 K(\lambda) =
 \begin{pmatrix}
a_0 + \lambda & r_0 \\
r_0 & -a_{0}+\lambda a_1 & \lambda r_{1} &  &    \\
  &\lambda r_{1} & a_2-\lambda a_{1} &  r_2 &   \\
   & &  r_2 & -a_2+\lambda a_3 &   \lambda r_3 &  \\
  &   &  & &  \ddots  \\
 \end{pmatrix}.
\]
Here, nondegenerate means that all off-diagonal entries of the matrix $K(\lambda)$ are nonzero
if $\lambda \ne 0$ and $|a_i| \ne 1$.

Hence one can define a family of formal monic orthogonal
polynomials $Q_n(x;\lambda)$ depending on the argument $x$ and an
additional parameter $\lambda$. These polynomials are uniquely defined
 through the 3-term recurrence relation
\be
Q_{n+1}(x;\lambda) + b_n(\lambda) Q_n(x;\lambda) + u_n(\lambda)
Q_{n-1}(x;\lambda) = x Q_n(x;\lambda), \lab{rec_Q_lambda}
\ee
where
\be b_n(\lambda) = \left\{a_n - \lambda a_{n-1} \quad \mbox{if}
\quad n \quad \mbox{is} \quad \mbox{even}  \atop \lambda a_n -
a_{n-1} \quad \mbox{if} \quad n \quad \mbox{is} \quad \mbox{odd}
\right . ,\lab{b_lambda} \ee and \be u_n(\lambda) =
\left\{\lambda^2 (1-a_{n-1}^2) \quad \mbox{if} \quad n \quad
\mbox{is} \quad \mbox{even} \atop  1-a_{n-1}^2 \quad \mbox{if}
\quad n \quad \mbox{is} \quad \mbox{odd}  \right . .\lab{u_lambda}
\ee Note that the eigenvalue problem for the orthogonal polynomials
$Q_n(x;\lambda)$ can be presented in algebraic form as \be (L +
\lambda M -x I) \vec q =0, \lab{lambda_x_q} \ee where $\vec q$ is a
vector constructed from the (non-monic) polynomials $Q_n(x;\lambda)$.
Equation \re{lambda_x_q} contains two parameters $\lambda$ and $x$
and belongs to the class of the so-called multi-parameter eigenvalue problems
\cite{Atkinson}, \cite{Sleeman}.

\begin{theorem}\label{LM_DrCh}
 Let $J=J(\lambda,x)=L+\lambda M-xI$. Then there exists a nontrivial choice of parameters $\t x$ and $\t\lambda$ such that
the Darboux transformation of $J$ subject to~\eqref{DK_lin} is of the form
\[
 \t J=\widetilde{J(\lambda,x)}=L+\t\lambda M-\t x I.
\]
Thus,  $J$ gives an explicit example of nontrivial 1-periodic dressing chains.
\end{theorem}
\begin{proof}
 For this special case, relation~\eqref{gen_quad} reduces to the following
\[
 \t\lambda r_1=s_2, \quad r_2=\lambda s_1, \quad \lambda s_0-xs_2=\t\lambda r_0-xr_2, \quad
r_0-\t x r_1=s_0-xs_1,
\]
\[
 r_1+\t \lambda r_2 -\t x r_0=s_1+\lambda s_2-x s_0.
\]
Putting $s_1=r_1=1$ we arrive at
\[
 \t x^2\lambda -\t x(x\lambda+x\t\lambda) +\lambda x^2=0,
\]
which is solvable. Moreover, it has a nontrivial solution (i.e. it is different from
the trivial solution $\t x=x$, $\t \lambda=\lambda$).
 \end{proof}
\begin{remark}
 Notice that  
\begin{equation}\label{AComR}
J(\lambda,0)J(-\lambda,0)-qJ(-\lambda,0)J(\lambda,0)=2(1-|\lambda|^2)I. 
\end{equation}
for $q=-1$. It is worth mentioning that some examples of the pencils of Jacobi matrices 
verifying~\eqref{AComR} for $q\in (-1,1)$ were given in~\cite{BW10} and~\cite{BMW}.
\end{remark}

\begin{remark}
Interestingly,~\eqref{AComR} becomes an anti-commutation relation at the same time as
the spectrum of $K(\lambda)=L+\lambda M$ becomes an interval (the spectrum 
consists of two disjoint intervals for $|\lambda|\ne 1$)~\cite{DVZh}. 
\end{remark}

Taking into account Theorem~\ref{LM_DrCh}, we see that the Jacobi matrices corresponding to the 
little and big -1 Jacobi polynomials~\cite{VZ_little},~\cite{VZ_big},~\cite{VZ_Bochner} give 
explicit examples of dressing chains.
The rest of the paper is devoted to constructing a realization related to the 
Bannai-Ito polynomials.

\begin{remark}
The main property of $J(\lambda,x)$ we used in the proof of Theorem~\ref{LM_DrCh} is the following one
\[
 L^2=M^2=I.
\]
 This property is based on the fact that $|a_j|\le 1$. However, it is easy to see how to modify the matrices $L$ and $M$
in order to preserve the above property for an arbitrary real sequence $a_j$. Indeed, let us introduce the block-diagonal matrix
\[
L =
 \begin{pmatrix}
  a_{0} & r_{0} &  &    \\
  \epsilon_0r_{0} & -a_{0} &  &   \\
   &  &              a_{2} & r_{2}  \\
   &  &             \epsilon_2 r_{2} & -a_{2}  \\
  &  &    & &          a_{4} & r_{4}  \\
   &  &    & &        \epsilon_4 r_{4} & - a_{4}  \\
&   &  & & & & \ddots  \\
 \end{pmatrix},
\]
and the block diagonal matrix 
\[
 M =
 \begin{pmatrix}
1 \\
& a_{1} & r_{1} &  &    \\
  &\epsilon_1r_{1} & -a_{1} &  &   \\
   & &  &               a_{3} & r_{3}  \\
   & & &             \epsilon_3 r_{3} & -a_{3}  \\
  &  & &   & &           a_{5} & r_{5}  \\
   &  & &   & &         \epsilon_5r_{5} & - a_{5}  \\
&   &  & & & & & \ddots  \\
 \end{pmatrix},
\]
where $r_j=\sqrt{|1-a_j^2|}$,  $\epsilon_j=-1$ if $|a_j|>1$, and $\epsilon_j=1$ if $|a_j|\le 1$.
Obviously, these new matrices $L$ and $M$ satisfy the main property : $L^2=M^2=I$.  
Furthermore, formulas~\eqref{rec_Q_lambda}, \eqref{b_lambda}, and~\eqref{u_lambda} remain the same.
However, in the general case, the matrix $K(\lambda)=L+\lambda M$ is not symmetric in a Hilbert space but 
is rather symmetric in a Krein space (for example, see~\cite{D09},~\cite{DD11}). 

Finally, note that the Szeg\H{o} orthogonal polynomials (satisfying~\eqref{Sz_rec}) with the 
condition $|a_j|>1$ were constructed in~\cite{VZh99}. Thus,
those polynomials are also related to dressing chains.
\end{remark}

\section{Schur-Delsarte-Genin maps}

In this section we give the precise definition of the Schur-Delsarte-Genin maps~\cite{DVZh}. 

In order to make~\eqref{rec_Q_lambda} more uniform, we will need the following special 
case of the Christoffel formula (see~\cite{Sz}).
\begin{lemma}\label{lemma1}
Let $P_n(x)$ be monic orthogonal polynomials satisfying the
recurrence relation
\be P_{n+1}(x) +(\theta - A_n - C_n)P_n(x) +
C_n A_{n-1} P_{n-1}(x) = xP_n(x), \lab{3-term_AC} \ee
with the standard initial conditions
\be
P_0=1, \; P_1(x) = x-\theta+A_0 .
\ee
Assume that the real coefficients $A_n,C_n$ are such that
$A_{n-1}>0,\; C_n
>0, \; n=1,2\dots$ and that $C_0=0$. In this case, for any real $\theta$ 
one has that $P_n(\theta)\ne 0$ for $n=1,2,\dots$.
Therefore, one can define the new polynomials \be \t P_n(x) = \frac{P_{n+1}(x) - A_n
P_n(x)}{x-\theta}. \lab{tP_CT1} \ee Then the monic polynomials $\t
P_n(x)$ are orthogonal and satisfy the recurrence relation \be \t
P_{n+1}(x) +(\theta - A_n - C_{n+1}) \t P_n(x) + C_n A_{n} \t
P_{n-1}(x) = x \t P_n(x). \lab{3-term_tP} \ee The inverse
transformation from the polynomials $\t P_n(x)$ to the polynomials
$P_n(x)$ is given by the formula \be P_n(x) = \t P_n(x) - C_n \t
P_{n-1}(x) . \lab{GT} \ee
\end{lemma}

To prove Lemma~\ref{lemma1} it is sufficient to observe that relation~\eqref{3-term_AC}
provides the explicit $LU$-factorization of the monic Jacobi matrix $J-\theta$
corresponding to the polynomials $P_n(x)$. The rest is an obvious reformulation
of the theory of spectral (or Darboux) transformations \cite{ZhS},
\cite{BM}.

Note also that \be A_n = \frac{P_{n+1}(\theta)}{P_n(\theta)}
\lab{A_PP} \ee which can easily be verified directly from
\re{3-term_AC}.

\begin{remark}
If the polynomials $P_n(x)$ are orthogonal with respect to a
weight function $w(x)$  then the polynomials $\t P_n(x)$ are orthogonal with respect to the
weight function $(x-\theta)w(x)$.
\end{remark}

Now, let us apply Lemma~\ref{lemma1} to $Q_n(x;\lambda)$ for   $\theta=\lambda+1$.
Namely,   we have that
\be
A_n= \left\{1-a_{n} \quad \mbox{if} \quad n \quad \mbox{is} \quad \mbox{even} \atop
\lambda(1-a_{n}) \quad \mbox{if} \quad n \quad \mbox{is} \quad \mbox{odd}  \right .,\quad
C_n= \left\{ \lambda(1+a_{n-1}) \quad \mbox{if} \quad n \quad \mbox{is} \quad \mbox{even} \atop
1+a_{n-1}\quad \mbox{if} \quad n \quad \mbox{is} \quad \mbox{odd}  \right ..
\lab{A_n_lambda_1} \ee
Then the transformed polynomials $\t Q_n(x;\lambda)$  satisfy the recurrence relation

\begin{equation}\label{UniF}
\t Q_{n+1}(x;\lambda)  +(-1)^n(\lambda-1) \t Q_{n}(x;\lambda) +
\lambda u_n^* \t Q_{n-1}(x;\lambda)=x \t Q_{n}(x;\lambda),
\end{equation}
where $u_n^*=(1+a_{n-1})(1-a_n)$ does not depend on $\lambda$.

It will sometimes be useful to renormalize the polynomials $\t Q_n(x;\lambda)$ as follows. 

\begin{lemma}\label{lemma3} Let the polynomials $\t Q_{n}(x;\lambda)$ satisfy~\eqref{UniF} and be orthogonal
with respect to a weight function $w_\lambda(x)$ on the set $E_\lambda$ for positive $\lambda$.
Then the monic polynomials defined via
\be
S_n(x;\lambda,\lambda_0)=(\sqrt{\lambda})^{-n}\t Q_n(\sqrt{\lambda}x;\lambda\lambda_0)
\ee
are orthogonal with respect to the weight function $w_\lambda(\sqrt{\lambda}x)$ on the set
$E_\lambda^*=\{x:\,\, \sqrt{\lambda}x\in E_\lambda\}$
and satisfy
\begin{equation}\label{UniF_1}
S_{n+1}(x;\lambda,\lambda_0)  +(-1)^n\chi S_{n}(x;\lambda,\lambda_0) +
\lambda_0u_n^* S_{n-1}(x;\lambda,\lambda_0)=x S_{n}(x;\lambda,\lambda_0),
\end{equation}
where $\chi=\lambda_0\sqrt{\lambda}-\frac{1}{\sqrt{\lambda}}$.

The converse is also true.
\end{lemma}

The proof is straightforward by making the substitutions $\lambda\to\lambda\lambda_0$ and  $x\to\sqrt{\lambda}x$
in~\eqref{UniF}.

The above statement is the final step in the construction of the SDG map.  
The map from the polynomials $\Phi_n$ orthogonal on the unit circle to the polynomials $S_n$ defined 
by~\eqref{UniF_1} is called {\bf the Schur-Delsarte-Genin map}. Note that this map reduces to the Delsarte-Genin 
map~\cite{DG1} when $\lambda=\lambda_0=1$ (see~\cite{DVZh},~\cite{Zhe1} for more details).
 
To complete this section, we provide a lemma which will be useful in the identification  of polynomial
systems with known families of orthogonal polynomials.

\begin{lemma}[\cite{Chi_Bol}]\label{lemma2}
Let $P_n(x)$ and $\t P_n(x)$ be two systems of orthogonal
polynomials defined as in the previous Lemma. Also, assume that
the polynomials $P_n(x)$ are orthogonal with respect to a
weight function $w(x)$ on the finite interval $[\alpha,\beta]$ and take any
$\theta\in(-\infty,\alpha)$. Finally, define the following monic polynomials
\be S_{2n}(x) = P_n(x^2+\alpha-c^2) ,
\quad S_{2n+1}(x) = (x-\chi) \t P_n(x^2+\alpha-c^2), \lab{S_PP} \ee
where $c$ is a positive number such that $\sqrt{\alpha-\theta}\le c$
and $\chi$ is a real number defined by the relation $\theta=\chi^2+\alpha-c^2$.
Then the polynomials $S_n(x)$ are orthogonal with respect to the weight function
\be
(\sign(x))(x+\chi)w(x^2+\alpha-c^2)
\ee
on the union of two intervals $[-\sqrt{\beta-\alpha+c^2},-c]\cup [c,\sqrt{\beta-\alpha+c^2}]$
and satisfy the recurrence
relation \be S_{n+1} + (-1)^n \chi S_n(x) + v_n S_{n-1}(x) =
xS_n(x), \lab{rec_S1} \ee
where $v_n$ are given by \be v_{2n} =-C_n, \; v_{2n+1}=-A_n.
\lab{v_CA} \ee
The converse statement is also true.
\end{lemma}

The orthogonality part of this lemma can be checked by straightforward computations.
Substituting~\eqref{S_PP} into the corresponding three-term recurrence relations,
we arrive at~\eqref{rec_S1} and~\eqref{v_CA}. For more details, see~\cite{Chi_Bol}.

\section{Complementary Bannai-Ito polynomials}

In this section we show that the complementary Bannai-Ito polynomials can be obtained 
from the Schur-Delsarte-Genin transformation of some polynomials.

For the reader's convenience we begin with the explicit formulas for
the general Bannai-Ito polynomials~\cite{BI}: \be P_{n+1}(x) + (\rho_1 - A_n -
C_n)P_n(x) + A_{n-1}C_n P_{n-1}(x) =x P_n(x), \lab{rec_AC_BIP} \ee
where \be A_n=\left\{ { {\frac { \left( n+1+2\,{\rho_1}-2\,{r_1}
\right) \left( n+1 +2\,{\rho_1}-2\,{ r_2} \right) }{4(n+1-{ r_1}-{
r_2}+{\rho_1}+ {\rho_2})}} , \quad \mbox{} \quad n \quad
\mbox{even} \atop {\frac { \left( n+1-2\,{r_1}-2\,{
r_2}+2\,{\rho_1}+2\,{ \rho_2} \right) \left(
n+1+2\,{\rho_1}+2\,{\rho_2} \right) }{4(n+
1-{r_1}-{r_2}+{\rho_1}+{\rho_2})}}
  ,\quad\mbox{} \quad n \quad
\mbox{odd}} \right . \lab{A_BIP} \ee and \be C_n=\left\{ {
-{\frac {n \left( n-2\,{r_1}-2\,{r_2} \right) }{4(n-{r_1}-
{r_2}+{\rho_1}+{\rho_2})}}
, \quad \mbox{} \quad n \quad
\mbox{even} \atop -{\frac { \left( n-2\,{r_2}+2\,{\rho_2} \right)  \left( n-2
\,{r_1}+2\,{\rho_2} \right) }{4(n-{r_1}-{r_2}+{\rho_1}+{
\rho_2})}},
 \quad
\mbox{} \quad n \quad \mbox{odd}.} \right  .\lab{C_BIP} \ee

We see that the Bannai-Ito polynomials admit the
representation~\eqref{3-term_AC} from Lemma~\ref{lemma1}. So, it is possible to introduce 
their companion polynomials~\cite{TVZh12}, which we
denote by $W_n(x)$, as follows \be W_n(x) = \frac{P_{n+1}(x) - A_n
P_n(x)}{x-\rho_1} \lab{W_P} \ee with $A_n$ given by \re{A_BIP}.
The polynomials $W_n(x)$ satisfy the recurrence relation
\re{3-term_tP}. In our special case, this relation reads 
\begin{equation}\lab{rec_W}
W_{n+1}(x) +(-1)^n \rho_2 W_n(x) + v_n W_{n-1}(x) = x W_n(x),
\end{equation} where the coefficients $v_n$ have the following explicit form 
\begin{equation}\lab{v_BI}
\begin{split}
v_{2n} &= -\frac{n(n+\rho_1-r_1+1/2)(n+\rho_1-r_2+1/2)(n-r_1-r_2)}{(2n+1+g)(2n+g)},\\
v_{2n+1}&=-\frac{(n+g+1)(n+\rho_1+\rho_2+1)(n+\rho_2-r_1+1/2)(n+\rho_2-r_2+1/2)}{(2n+1+g)(2n+g+2)},
\end{split}
\end{equation} 
with $g=\rho_1+\rho_2-r_1-r_2$.

One can observe that the companion polynomials $W_n$ defined by~\eqref{rec_W} and the polynomials $S_n$ defined 
by~\eqref{UniF_1} are, in fact, of the same nature. So, we can therefore try to identify them. 
First, let us write the equality
\begin{equation}\label{fromQtoW1}
v_n=\lambda_0u_n^* =\lambda_0(1+a_{n-1})(1-a_{n}),
\end{equation}
for $n=0,1,2,\dots$. From~\eqref{fromQtoW1} we see that the parameters $a_n$ can be determined as follows
\begin{equation}
a_n=1-\frac{v_n}{\lambda_0(1+a_{n-1})},
\end{equation}
and  $a_0$, $\lambda_0$ are  arbitrary parameters. Setting
\[
 a_0=1-2\frac{(\rho_1-r_2+1/2)(-r_1-r_2)}{(-r_2-\rho_2-1/2)(g+1)}, \quad 
\lambda_0=\frac{1}{4}(-\rho_2-r_2-1/2)(\rho_2-r_2-1/2)
\]
we arrive at
\be
a_n = \left\{ 1-2\,{\frac { \left( \rho_1-r_2+(n+1)/2\right)  \left( -r_1-r_2+n/2 \right) }
{ \left( -\rho_2-r_2-1/2 \right)  \left( n+g+1 \right) }}, \quad n
\quad \mbox{even}    \atop   
1-2\,{\frac {\left( g+(n+1)/2\right)  \left( \rho_2-r_2+n/2 \right) }{ \left( \rho_2-r_2-1/2 \right)  \left( n+g+1 \right) }},
\quad n \quad \mbox{odd.}   \right. \lab{a_Bannai} \ee

To complete the identification, it remains to take 
\begin{equation}\label{lforBI}
\lambda=\lambda_{BI}=\frac{4}{(-\rho_2-r_2-1/2)^2}
\end{equation} 
to ensure that  
\[
\lambda_0\sqrt{\lambda_{BI}}-\frac{1}{\sqrt{\lambda_{BI}}}=\rho_2.
\]

\begin{remark}
First, notice that 
\begin{equation}\label{QRforK}
\begin{split}
K^2(\lambda\lambda_0)=&(L+\lambda\lambda_0M)(L+\lambda\lambda_0M)=
(1-\lambda\lambda_0)^2I+\lambda(L+M)^2\\
=&\frac{1}{\lambda}\left[(\lambda_0\sqrt{\lambda}-\frac{1}{\sqrt{\lambda}})^2I+\lambda_0(L+M)^2\right]\\
=&\frac{1}{\lambda}\left[\chi^2I+\lambda_0(L+M)^2\right], \quad \lambda\ne 0.
\end{split}
\end{equation}
From~\eqref{QRforK} and the spectral mapping theorem we see that $\sigma(K(\cdot))$ consists of
numbers that are the square root of a quadratic polynomial in $\lambda$.
Since it is known that the orthogonality measure for the Bannai-Ito polynomials 
(and, thus, for the complementary Bannai-Ito polynomials) has a uniform grid of mass points~\cite{TVZh12}, 
one can conclude that, for the point $\lambda=\lambda_0\lambda_{BI}$ and the coefficients 
$a_n$ defined in~\eqref{a_Bannai}, the quadratic polynomial $\chi^2+\lambda_0\mu_n^2$, where
$\mu_n$ is an eigenvalue of $L+M$, becomes an exact square. 
\end{remark}

%%%%%%%%%%%%%%%%%%%%%%%%%%%%%%%%%%%%%%%%%%%%%%%%%%%%%%%%%%%%%%%%%%%%%%%%%%%%%%%%%%%%%%%%%%%%%%%%%%%%%%%%%%%

\section{Racah-Wilson polynomials on the unit circle}

We see from~\eqref{QRforK} that in order to characterize the spectrum of $K(\lambda)$ we need
to know that spectrum precisely, at one point $\lambda$ at least. 
In this section, we identify such a determining point and show that this point is 
the preimage of the SDG map corresponding to the symmetrized Racah-Wilson polynomials.

 Recall~\cite{KS} that the Wilson polynomials $\t W_n(x)$ are defined through the 3-term recurrence relation
\be
\t W_{n+1}(x) + (A_n+C_n-\beta_1^2) \t W_n(x) + A_{n-1}C_n \t W_{n-1}(x) = x\t W_n(x) \lab{recW} \ee  
where $\sigma = \beta_1+\beta_2+\beta_3+\beta_4$ and the coefficients $A_n$, $C_n$ are as follows
\be
A_n=\frac{(n+ \sigma-1)(n+\beta_1+\beta_2)(n+\beta_1+\beta_3)(n+\beta_1+\beta_4)}{(2n+\sigma-1)(2n+\sigma)}, \lab{A_W} \ee
\be
C_n=\frac{n(n+\beta_2+\beta_3-1)(n+\beta_2+\beta_4-1)(n+\beta_3+\beta_4-1)}{(2n+\sigma-1)(2n+\sigma-2)}. \lab{C_W} \ee
We shall refer to the above polynomials as the Racah-Wilson polynomials 
since for the discrete measure case (which is essential in the study of the Bannai-Ito polynomials), they represent the same
type of polynomials.

Regarding the aim of this section, let us first note that it was already shown in~\cite{TVZh12} 
that the complementary Bannai-Ito polynomials can be constructed from 
the Racah-Wilson polynomials. Namely, the complementary Bannai-Ito polynomials are obtained from the Racah-Wilson polynomials 
by applying the Chihara construction (in other words, by using Lemma~\ref{lemma2}).

We can repeat this construction for our purpose but it is in fact enough to observe that 
setting $\lambda=1/{\lambda_0}$ at the last stage of the identification in the previous section leads 
to the following symmetric polynomials
\be
S_{n+1}(x;\lambda_0) + v_n S_{n-1}(x;\lambda_0) = xS_n(x;\lambda_0). \lab{rec_SW} \ee

 Clearly, the polynomials $i^nS_n(ix;\lambda_0)$ are symmetrized Racah-Wilson polynomials.
More precisely, the polynomials $i^nS_n(ix;\lambda_0)$ can be obtained from $\t W_n(x)$ using Lemma~\ref{lemma2}
with $\chi=0$ as desired. Obviously, we have
\[
\beta_1=\rho_2,\quad \beta_2=1+\rho_1,\quad \beta_3=-r_1+\frac{1}{2},\quad \beta_4=-r_2+\frac{1}{2}.
\]
It should be stressed that the choice $\chi=0$ (that is $\lambda\lambda_0=1$ and 
the corresponding Jacobi matrix being $K(1)=L+M$) corresponds to the classical Delsarte-Genin map.
Therefore, the orthogonal polynomials associated with the reflection parameters  
\be
a_n = \left\{ 1-2\,{\frac { \left( \beta_{{2}}+\beta_{{4}}+n/2-1 \right)\left(\beta_{{3}}+\beta_{{4}}+n/2-1\right)}
{\left( \beta_{{4}}-\beta_{{1}}-1\right)  \left( n+\sigma -1\right) }}, \quad n
\quad \mbox{even}    \atop   1-2\,{\frac { \left( \sigma+n/2-3/2 \right)\left( \beta_{{1}}+\beta_{{4}}+n/2-1/2 \right)}
{ \left( \beta_{{1}}+\beta_{{4}} -1\right)  \left( n+\sigma-1 \right) }},
\quad n \quad \mbox{odd}   \right. \lab{a_Wilson} \ee
are the preimages of the symmetrized Racah-Wilson polynomials under the Delsarte-Genin map.

%%%%%%%%%%%%%%%%%%%%%%%%%%%%%%%%%%%%%%%%%%%%%%%%%%%%%%%%%%%%

\section{Bannai-Ito polynomials}

In this section we show how the Bannai-Ito polynomials appear in the framework of the Schur-Delsarte-Genin maps.

First, notice that the complementary Bannai-Ito polynomials are obtained from the Bannai-Ito polynomials through a
Christoffel transformation. Actually, the same can be said about the polynomials $S_n$ and $Q_n$  defined 
by~\eqref{UniF_1} and~\eqref{rec_Q_lambda}, respectively. Since we have already identified $S_n$ with
the complementary Bannai-Ito polynomials, one could expect that there is a similar relation between the Bannai-Ito 
polynomials and the polynomials $Q_n$.
However, it is easy to see from~\eqref{A_n_lambda_1}, ~\eqref{A_BIP},~\eqref{C_BIP}, and~\eqref{a_Bannai} that 
the polynomials $Q_n$ cannot be identified with the Bannai-Ito polynomials.  Generally speaking, the reason is 
that they are related by a Uvarov transformation and do not coincide. Nevertheless, the Uvarov transformation 
can be realized via the following substitution
\begin{equation}\label{Subs1}
\rho_1\mapsto -r_2-\frac{1}{2},\quad \rho_2\mapsto\rho_1,\quad 
r_1\mapsto r_1,\quad r_2\mapsto -\rho_2-\frac{1}{2}. 
\end{equation}

Indeed, applying~\eqref{Subs1} to~\eqref{a_Bannai} gives

\be
a_n = \left\{ 1-2\,{\frac { \left( \rho_1-r_2+(n+1)/2\right)  \left( \rho_1-r_1+(n+1)/2 \right) }
{ \left( \rho_1-\rho_2\right)  \left( n+g+1 \right) }}, \quad n
\quad \mbox{even}    \atop   
1-2\,{\frac {\left( g+(n+1)/2\right)  \left( \rho_1+\rho_2+(n+1)/2 \right) }{ \left( \rho_1+\rho_2 \right)  \left( n+g+1 \right) }},
\quad n \quad \mbox{odd.}   \right. \lab{a_Bannai_1} \ee

Next, setting 
\[
 \lambda=\lambda_0\lambda_{BI}=\frac{\rho_2+\rho_1}{\rho_2-\rho_1}
\]
the calculation of $b_n(\lambda)$ and $u_n(\lambda)$ (see formulas~\eqref{b_lambda} and~\eqref{u_lambda}) 
for these new coefficients $a_n$ leads to the following equalities
\begin{equation}\label{QtoBI}
 b_n(\lambda)=\frac{2}{\rho_2-\rho_1}(\rho_1-A_n-C_n),\quad
u_n(\lambda)=\frac{4}{(\rho_2-\rho_1)^2}A_{n-1}C_n,
\end{equation}
where $A_n$ and $C_n$ are defined by~\eqref{A_BIP} and~\eqref{C_BIP}, respectively.

According to Lemma~\ref{lemma3} and~\eqref{QtoBI}, we may see that the polynomials $Q_n(x; \lambda)$ 
and the Bannai-Ito polynomials~\eqref{rec_AC_BIP} coincide up to a renormalization.

\begin{remark}
To explain the choice of $\lambda=\lambda_0\lambda_{BI}$, notice that the complementary Bannai-Ito polynomials are
the Christoffel transforms of the Bannai-Ito polynomials and that this $\lambda$ was picked for the SDG map to give 
the complementary Bannai-Ito polynomials. Thus, if we want to get the Bannai-Ito polynomials as the polynomials $Q_n$,
the same $\lambda$ has to be taken under the transformation~\eqref{Subs1}. 

It is also natural from the point of view of the mass points of the orthogonality measure. 
Indeed, the sets of mass points for the Bannai-Ito polynomials and the complementary Bannai-Ito polynomials 
should be the same except possibly for one point. 
\end{remark}

Finally, let us  summarize the results on the Bannai-Ito polynomials.
\begin{theorem}
 The following statements hold true:
\begin{enumerate}
 \item The complementary Bannai-Ito polynomials are the SDG images of the Szeg\H{o} polynomials corresponding
to the parameters~\eqref{a_Bannai}.
\item The Bannai-Ito polynomials coincide with the polynomials $Q_n$ defined by~\eqref{rec_Q_lambda} and 
in this case the parameters $a_j$ are given by~\eqref{a_Bannai_1}.
\item The difference between~\eqref{a_Bannai} and~\eqref{a_Bannai_1} is the
transformation~\eqref{Subs1}. 
\end{enumerate}

\end{theorem}

\bigskip\bigskip
{\Large\bf Acknowledgments}

\bigskip

M.D. and A. Zh. thank the University of Montreal for its hospitality in the course of this study. 
The research of LV is supported in part by a research grant from the Natural Sciences and Engineering Research 
Council (NSERC) of Canada.

\end{document}